\documentclass[10pt,letterpaper]{article}
\usepackage[margin=1in]{geometry}
\usepackage{amsmath,amssymb,amsthm}
\usepackage{float}
\usepackage{graphicx}
\usepackage{hyperref}
\usepackage[utf8]{inputenc}
\usepackage{xcolor}
\usepackage{palatino}
\usepackage{caption}
\usepackage{subcaption}
\usepackage{booktabs}
\usepackage{enumitem}
\usepackage{comment}
\usepackage{lipsum}
\usepackage{mathtools}
\usepackage{cuted}

\newcommand{\rl}[1]{\noindent{ \textcolor{orange}{\footnotesize[{\bf Ling}: #1]}}}
\newcommand{\dg}[1]{{\color{blue}{#1}}}

\newtheorem{theorem}{Theorem}
\newtheorem{definition}{Definition}
\newtheorem{lemma}[definition]{Lemma}
\newtheorem{corollary}[definition]{Corollary}
\newtheorem{condition}[definition]{Condition}

\newcommand{\cappa}{k} 

\begin{document}
\title{Bitcoin's Latency--Security Analysis Made Simple}
\author{
Dongning Guo \\ Northwestern University \\ dguo@northwestern.edu  \and Ling Ren \\ University of Illinois at Urbana-Champaign \\ renling@illinois.edu }
\date{}

\maketitle

\begin{abstract}
Simple closed-form upper and lower bounds are developed for the security of the Nakamoto consensus as a function of the confirmation depth, the honest and adversarial block mining rates, and an upper bound on the block propagation delay.
The bounds are exponential in the confirmation depth and apply regardless of the adversary's attack strategy.
The gap between the upper and lower bounds is small for Bitcoin's parameters.
For example, assuming an average block interval of ten minutes, a network delay bound of ten seconds, and 10\% adversarial mining power, the widely used 6-block confirmation rule yields a safety violation between 0.11\% and 0.35\% probability.  
\end{abstract}

\sloppy
\pagestyle{plain}

\newcommand{\f}[1]{f_{#1}}
\newcommand{\bl}[1]{b_{#1}}
\newcommand{\F}[2]{F_{#1,#2}}
\newcommand{\Y}[2]{Y_{#1,#2}}
\newcommand{\Z}[2]{Y_{#1,#2}}
\renewcommand{\H}[2]{H_{#1,#2}}
\newcommand{\A}[2]{A_{#1,#2}}
\newcommand{\D}[2]{D_{#1,#2}}
\newcommand{\Exp}{\mathsf E}
\newcommand{\expect}[1]{{\Exp}\left\{#1\right\}}

\section{Introduction}

The famed Bitcoin white paper presented a novel Byzantine fault tolerant consensus algorithm that is now known as the Nakamoto consensus~\cite{nakamoto2008bitcoin}.
A notable property of the Nakamoto consensus is that the deeper a transaction is in a longest blockchain,
the safer it is to commit the transaction.  
In fact, 
the probability that a transaction's safety is violated decreases essentially exponentially with the number of ``confirmations''.
A latency-security analysis of the Nakamoto consensus aims at deriving upper bounds on the safety violation probability of a given confirmation rule, under all possible attacks allowed in a formally described model. 

While the Nakamoto consensus protocol is simple and elegant, rigorously analyzing its latency and security turns out to be quite challenging.
The Bitcoin white paper did not provide a formal model or rigorous analysis;
instead, it only considered one specific attack. 
Garay et al.~\cite{garay2015bitcoin} provided the first latency-security analysis for the Nakamoto consensus against all possible attacks. Follow-up works extended their analysis from a lockstep synchrony model to the more realistic non-lockstep synchrony model~\cite{pass2017analysis,pass2017rethinking,kiffer2018better,ren2019analysis}. 
The fundamental fault tolerance limit of the Nakamoto consensus has also been obtained~\cite{dembo2020everything,gavzi2020tight}.


Most existing latency-security analysis of the Nakamoto consensus focused on asymptotic results. 
This means they showed that transactions in the Nakamoto consensus \emph{eventually} become permanent (also called committed, decided, finalized, or immutable in the literature), but do not provide concrete results on when that happens.  
Several recent works give concrete analysis~\cite{li2020liveness, li2021close, gavzi2021practical} 
but their methods are quite complex. 
In particular, \cite{li2020liveness,li2021close} resort to a stochastic analysis of races between renewal processes and focus on confirmation by time rather than the more practically relevant confirmation by depth.
Reference~\cite{gavzi2021practical} involves an complicated analysis of forkable strings and an iterative dynamic programming algorithm to numerically evaluate a Markov chain, and the results are not in a closed form. 

In this paper, we develop two sets of simple upper and lower bounds on the security of the Nakamoto consensus as a function of the confirmation depth, the mining rates, and a block propagation delay upper bound.  
One set of upper and lower bounds are essentially exponential functions in the confirmation depth.  A second set of closed-form bounds are numerically closer and still easy to compute using  finite sums of geometric and binomial distribution functions.

The gap between the upper and lower bounds is small for the Bitcoin parameters.
For example, assuming on average one block is mined every ten minutes, the delay bound is ten seconds, and 10\% of the mining power is adversarial, the probability of safety violation of the 6-block rule of thumb is bounded between 0.11\% and 0.35\%.  If the adversarial mining power increases to 25\%, one needs a 20-block rule to achieve similar bounds.
For the Bitcoin parameters and relatively small confirmation depths, our new upper bound is tighter than the best existing result in~\cite{gavzi2021practical}. 

In this paper, we introduce a new reduction technique to analyze Nakamoto consensus protocols.  We describe a ``rigged'' model in which some selected blocks mined by honest nodes are converted into adversarial blocks.  Such conversion can only make the adversary more powerful.  We judiciously select those blocks for conversion such that the sequence of honest/adversarial attributes of all blocks remains memoryless in the rigged model.
We 
show that the well-understood private-mining attack is a {\em best} attack in the rigged model.  Finally, we obtain simple upper and lower bounds
on the safety violation probability for the
original model by analyzing the success probability of the private-mining attack in the rigged model.  


The remainder of the paper is organized as follows.  In Section~\ref{sec:model:formal}, we present the canonical
model for the Nakamoto consensus.  The main theorems are given in Section~\ref{s:main}.  In Section~\ref{s:optimal}, we 
show that the well-known private-mining attack is optimal under a special condition.  In Section~\ref{s:probabilistic}, we construct a rigged model which on one hand makes the adversary more powerful, and on the other hand makes the private-mining attack optimal. Two sets of upper and lower bounds on the probability of safety violation are then obtained in Sections~\ref{s:probabilistic} and~\ref{s:asymptotics}. 
Numerical results are given in Section~\ref{s:numerical}.

\section{The Canonical Model}
\label{sec:model:formal}

We assume the readers are familiar with how the Nakamoto consensus protocol works.
We briefly describe the protocol below only to introduce notation.  
The protocol builds a growing sequence of transaction-carrying blocks where every block is chained to its predecessor block by a solution to a computational puzzle (i.e., a proof of work).  
At any point in time, every honest node attempts to ``mine'' a new block that extends \dg{a} 
longest chain of blocks (blockchain) to its knowledge; once a new longest blockchain is mined or received, an honest node sends it to other nodes through a gossip network.
We will use chains and blockchains interchangeably in this paper.
A node commits a block when 
at least $\cappa-1$ blocks are mined on top of it 
as part of a longest blockchain known to that node, where $\cappa$ is a natural number
chosen by the node.
We call this the {\em $\cappa$-confirmation commit rule}.

We now define the blockchain data structure.
Let us number all the blocks in the time order they are mined.
The $j$-th block in this numbering is called block $j$ for short.
We denote a blockchain using a sequence of block numbers, e.g., in blockchain $(b_0, b_1, b_2, \ldots, b_m)$, the $i$-th block in the chain is block $b_i$.
A blockchain always starts with the
Genesis block $b_0=0$ that is known to all nodes by time 0, when the protocol starts. 
Each subsequent block $b_i$ must contain a proof of work
that binds it to the predecessor block $b_{i-1}$.
The last block, $b_m$ in the above example, uniquely identifies the entire blockchain, so we will also refer to the above blockchain as blockchain $b_m$ or chain $b_m$.
The position of a block in the blockchain is called its \emph{height}.
We will use $h_b$ to denote the height of block $b$.
In the above example, block $b_i$ is on height $i$, i.e., $h_{b_i} = i$ (the Genesis block is on height 0).

We adopt a natural continuous-time model and model proof-of-work mining as a homogeneous 
Poisson point process.
Let $\lambda$ be the total mining rate of the network (honest and adversarial nodes combined).
Let $\rho \in (\frac12, 1]$ be the fraction of honest mining rate and $1-\rho$ be the fraction of adversarial mining rate.
Note that the model allows some of the nodes to have zero mining rate, hence capturing light nodes.
In the canonical model, a block is said to be honest (resp.\ adversarial) if it is mined by an honest (resp.\ adversarial) node.
One can think of honest block arrivals and adversarial blocks as two independent Poisson processes with rates $\rho\lambda$ and $(1-\rho)\lambda$, respectively.
Due to the Poisson merging and splitting properties, it is equivalent to think of all blocks arrivals as a single Poisson process with rate $\lambda$ where each block is honest with probability $\rho$ and adversarial with probability $1-\rho$.

We use $t_j$ to denote the time block $j$ is mined.
Evidently, $0=t_0\le t_1\le t_2\le \dots$.
If block $j$ is mined by an honest node, it must extend a longest blockchain to that node's knowledge immediately before $t_j$.
Ties are broken arbitrarily or by the adversary at will. 
The honest node will also immediately publish the newly mined block $j$ through a gossip network.

Without loss of generality,
we assume a single adversary controls all adversarial mining power.
If the adversary mines block $j$, the block may extend any blockchain mined by time $t_j$ and may be presented to each individual honest node 
at any time from $t_j$ onward at will. 
%
The adversary in practice cannot predict the arrival times of future blocks but our results hold even against an omniscient adversary
that sees the arrival times of all future blocks. 

We assume the standard (non-lockstep) synchrony model. 
We abstract away the topology and operations of the gossip network by assuming a universal block propagation delay upper bound, denoted as $\Delta\ge0$.
Applying it to the Nakamoto consensus, if any honest node mines or receives a new longest blockchain $b$ at time $t$, then all nodes receive blockchain $b$ by time $t+\Delta$.

A block may contain an arbitrary number of transactions. 
A transaction may appear in multiple chains but
it can appear at most once in any given chain.
The adversary's goal is to attack the safety of a target transaction~\cite{nakamoto2008bitcoin}, as defined below. 

\begin{definition}[Safety violation of a target transaction]
A target transaction's safety is violated if a block containing the target transaction is committed and a different
block (which may or may not contain the target transaction) is also committed on the same height.
\end{definition}

In this paper, we assume the target transaction appears in every node's view at time $\tau$. We also assume all honest nodes adopt the $\cappa$-confirmation commit rule with the same $\cappa$.
Our goal is to obtain tight bounds on the probability that the adversary violates the safety of the target transaction.

Table \ref{t:notation} illustrates frequently used notations in this paper. 

\begin{table}[tb]
\begin{center}
\begin{tabular}{ |c|l| } 
\hline
$\cappa$ & confirmation depth \\
$\lambda$ & total mining rate \\ 
$\rho$ & fraction of honest mining rate \\
$\Delta$ & block propagation delay upper bound \\
$h_b$ & height of block $b$ \\
$t_b$ & mining time of block $b$ \\
$p=\rho e^{-\lambda\Delta}$ & fraction of honest blocks \\ & \qquad in the rigged model \\
\hline
\end{tabular}
\end{center}
\smallskip
\caption{Some frequently used notations.}
\vspace{-8pt}
\label{t:notation}
\end{table}

\section{Main Results}
\label{s:main}

\begin{theorem} \label{th:exponential}
    Given the confirmation depth $\cappa$ as a natural number, the delay bound $\Delta \ge0$, the total mining rate $\lambda>0$, and the fraction of honest mining power~$\rho\in(0,1]$, as long as
    \begin{align} \label{eq:tolerance}
        p = \rho e^{-\lambda \Delta} > \frac12,
    \end{align}
    a target transaction's safety can be violated with probability no greater than
    \begin{align} \label{eq:aupper}
        \left( 2 + 2 \sqrt{\frac{p}{1-p}} \right) (4p(1-p))^{{\cappa}}
    \end{align}
    regardless of the adversary's attack strategy.

    Conversely, there exists an attack that violates the target transaction's safety with probability at least
    \begin{align} \label{eq:alower}
        \frac1{\sqrt{{\cappa}}} 
        \left( 4 \rho (1-\rho) \right)^{{\cappa}} .
    \end{align}
\end{theorem}

Interestingly, both the upper bound~\eqref{eq:aupper} and the lower bound~\eqref{eq:alower} can be expressed in the exponential form using a simple equality:
\begin{align}
    (4p(1-p))^{{\cappa}}
    =
    e^{-2 {\cappa} d\left( \frac12 \,\|\, p \right)}
\end{align}
where 
$    d(p \,\|\, q)
    =
    p \log \frac{p}q + (1-p) \log\frac{1-p}{1-q}
$
denotes the relative entropy between the Bernoulli($p$) and the Bernoulli($q$) distributions.  
As we shall see, the safety violation event is tied in some ways to the event that $\cappa$ or more out of $2\cappa$ independent Bernoulli trials are successful. For large $\cappa$, the probability of this event decays exponentially in $\cappa$, where the exponent takes the form of a relative entropy,
as is expected from the perspective of large deviations.  We emphasize that the bounds in Theorem~\ref{th:exponential} apply to all adversarial mining strategies and all choices of the confirmation depth.


We also provide another pair of upper and lower security bounds that are tighter than the bounds in Theorem~\ref{th:exponential}. The second set of bounds take somewhat more complicated form as finite series sums, but are still very easy to evaluate using the probability mass function (pmf) and cumulative distribution function (cdf) of the binomial and geometric distributions. Specifically, we use the following variant of the geometric distribution with parameter $p=1-q>\frac12$. For $i=1,2,\dots$,
its pmf is expressed as
\newcommand{\qop}{\frac{q}{p}}
\newcommand{\boa}{\frac{\beta}{\alpha}}
\begin{align}    \label{eq:P1}
    P_1(i;p) = \left( \qop \right)^{i-1} \left( 1 - \qop \right) ,
\end{align}
and its complementary cdf is expressed as
\begin{align} \label{eq:F1}
    \overline{F}_1(i;p) = \left(\qop\right)^{i} .
\end{align}
We also denote the pmf of the binomial distribution with parameters $(n,q)$ as
\begin{align} \label{eq:P2}
    P_2(j; n, q) = \binom{n}{j} q^j (1-q)^{n-j}
\end{align}
and the corresponding complementary cdf as
\begin{align} \label{eq:F2}
    \overline{F}_2(j; n,q) = \sum_{l=j+1}^n P_2(l;n, q) .
\end{align}

\begin{theorem} \label{th:bounds}
    Given the confirmation depth $\cappa$ as a natural number, the delay bound $\Delta\ge0$, the total mining rate $\lambda>0$, and the fraction of honest mining power~$\rho\in(0,1]$,
    as long as $p=\rho e^{-\lambda \Delta}>\frac12$,
    a target transaction's safety can be violated with probability no greater than
    \begin{align}
    \begin{split}
    \overline{F}_1(\cappa;p) + \sum_{i=1}^{\cappa} P_1(i;p) \cdot
    \bigg(\overline{F}_2(\cappa-i; 2{\cappa}+1-i,1-p) 
    + \sum_{j=0}^{\cappa-i} P_2(j; 2{\cappa}+1-i,1-p) \cdot  \overline{F}_1(2\cappa+1-2i-2j;p) \bigg)
    \end{split}
    \label{eq:upperbound}
    \end{align}
    regardless of the adversary's attack strategy.
    
    Conversely, there exists an attack that violates the target transaction's safety with probability at least
    \begin{align}
    \begin{split}
    \overline{F}_1(\cappa;\rho) + \sum_{i=1}^{\cappa} P_1(i;\rho) \cdot
    \bigg(\overline{F}_2(\cappa-i; 2{\cappa}+1-i,1-\rho) 
    + \sum_{j=0}^{\cappa-i} P_2(j; 2{\cappa}+1-i,1-\rho) \cdot  \overline{F}_1(2\cappa+2-2i-2j;\rho) \bigg) .
    \end{split}
    \label{eq:lowerbound}
    \end{align}
\end{theorem}

The proofs of these bounds are relegated to Sections~\ref{s:optimal}--\ref{s:asymptotics}.
The exponential bounds in Theorem~\ref{th:exponential} can be thought of as the large deviations approximations of those in Theorem~\ref{th:bounds}.
As we shall see in Section~\ref{s:numerical}, the bounds in Theorem~\ref{th:bounds} are numerically closer than the bounds in Theorem~\ref{th:exponential}.


\newcommand{\safetyII}{type-II safety}

\section{The Private-Mining Attack Is Conditionally Optimal}
\label{s:optimal}

Let us introduce the simple private-mining attack. Its general structure was mentioned explicitly in~\cite{sompolinsky2016bitcoin} and even earlier works. 

\begin{definition}[Private-mining attack against a target transaction $tx$]
    Starting from time 0, every adversarial block extends a longest blockchain that does not contain $tx$.
    The propagation of every honest block is maximally delayed (i.e., by $\Delta$).
    All adversarial blocks are kept private until the the adversary can violate the safety of $tx$ by publishing all blocks.
\end{definition}

Throughout this section, we assume
the following simple condition on block heights is always upheld: 
\begin{condition}
All honest blocks 
are on different heights.
Also, honest blocks mined after time $\tau$ (when the target transaction appears) are higher than honest blocks mined by time $\tau$. 
\label{honest_condition}
\end{condition}

Condition~\ref{honest_condition} always holds if the delay bound $\Delta=0$. 
With $\Delta>0$, honest nodes 
may mine multiple blocks on the same height.
But Condition~\ref{honest_condition} can be upheld by keeping no more than one honest block on each height. This changes the mining statistics and we relegate this discussion to Section~\ref{s:probabilistic}, where we use this technique along with a reduction argument to bound the safety violation probability.
In this section, we will prove that under Condition~\ref{honest_condition}, the private-mining attack is one best attack in the following sense:



\begin{theorem}
    Under Condition~\ref{honest_condition},
    if any attack succeeds in violating the target transaction's safety, then the private-mining attack also succeeds in violating the target transaction's safety.
\label{thm:best_strategy}
\end{theorem}

Theorem~\ref{thm:best_strategy} holds for any given honest and adversarial block mining times;
in other words, as long as Condition~\ref{honest_condition} holds, the private-mining attack is a best attack for every {\em sample path}, regardless of statistics.
It was claimed in~\cite[Appendix F]{dembo2020everything} that the private-mining attack is optimal in violating the safety of a predetermined target block in the special case of $\Delta=0$, but the proof therein does not apply to the case when the target block is mined by an adversary.
Thus, the optimality of the private-mining attack has not been fully established thus far even in the case of $\Delta=0$.


Before proving the theorem, we establish some additional terminology and simple facts.
The private-mining attack consists of two stages. 
In the first stage, between time $0$ to $\tau$, the adversary tries to build a ``lead'', formally defined as follows:

\begin{definition}[lead]
    The lead (of the adversary) at  time $t$ is the height of the highest block mined by time $t$ minus the height of the highest honest block mined by $t$.
\end{definition}

By definition, the lead is never negative.
In the private-mining attack up to time $\tau$, if a highest (private) adversarial block is higher than any honest block (the lead is positive), the adversary mines on this highest adversarial block to try to increase the lead; 
otherwise, the lead is zero, and the adversary mines on a highest honest block to try to obtain a positive lead. 
We can show that this strategy maximizes the lead.

\begin{lemma} \label{lem:maxlead}
Under condition~\ref{honest_condition}, the private-mining attack maximizes the lead at all times up to $\tau$.
\end{lemma}
\begin{proof}
    Let $l_t$ denote the lead at time $t$.  The lead may only change upon block arrivals.
    Upon the mining of every adversarial block, the lead advances by at most 1.
    Upon the mining of every honest block, the lead decreases by at least 1 unless it stays 0.
    Under the private-mining attack up to $\tau$, the lead advances by exactly 1 upon the mining of every adversarial block; the lead decreases by exactly 1 unless it stays 0 upon the mining of every honest block.
    Hence, the private-mining attack achieves the maximum lead at all times up to $\tau$.
\end{proof}

The second stage of the private-mining attack starts at time $\tau$. 
Honest nodes will include $tx$ in the next block on the honest chain.
The adversary tries to build a private chain that does not contain the target transaction $tx$. If there ever comes a time after an honest node commits the target transaction, that the adversary's private chain is as long as the public chain, 
then the adversary publishes its private chain and the attack succeeds in violating the safety of $tx$. 
If such an instance never occurs, then the private-mining attack on $tx$ fails. 

For convenience, we define the following notions: 

\begin{definition}[public]
\label{def:pubic}
A blockchain is public at time $t$ if it is included in all honest nodes' views at time $t$. We say block $b$ is public if and only if blockchain $b$ is public.
\end{definition}

\begin{definition}[credible]
\label{def:credibility}
A blockchain is credible at time $t$ if it is no shorter than any public blockchain at time $t$.
\end{definition}
Equivalently, a credible blockchain at time $t$ must be no shorter than at least one honest node's longest blockchain at time $t$. 
Under the $\cappa$-confirmation rule, a credible blockchain can be used to convince at least one honest node to commit blocks that are $\cappa$ deep in this blockchain.
Furthermore, an honest node attempts to extend only
credible blockchains. 
A block mined by an honest node must be credible at its mining time;
it then loses its credibility at a later time, and cannot regain credibility afterward.

We are now ready to prove Theorem~\ref{thm:best_strategy}.

\begin{proof}[Proof of Theorem~\ref{thm:best_strategy}]
Let us first consider the hypothetical attack that violates the safety of the target transaction $tx$.
Let block $c$ and $d$ be
a 
pair of blocks that minimize $h_+ = \max(h_c,h_d)$ and satisfy the following conditions (see Figure~\ref{fig:hah+} for illustrations):
\begin{enumerate}
    \item Blockchain $c$ is credible at time $t_c$, contains the target transaction in a block $b$, and has height $h_c \geq h_b + {\cappa-1}$, and 
    
    \item Blockchain $d$ is credible at time $t_d$, has height $h_d \geq h_b + {\cappa-1}$, does not contain block $b$, 
    and does not contain the target transaction on heights up to $h_b-1$.
\end{enumerate}
Blocks $c$ and $d$ as defined must exist in order for some honest node to commit the target transaction in a block and some honest node (possibly the same one) to commit a different block on the same height.
Since chains $c$ and $d$ are credible at their respective mining times, and higher and higher honest blocks are mined over time, such a pair is determined by the time a block on some height greater than $h_+$ becomes public.

Since $h_c \geq h_b+{\cappa-1}$ and $h_d \geq h_b+{\cappa-1}$, we have
\begin{equation}
h_+ = \max(h_c, h_d) \geq h_b + {\cappa-1}.
\label{eqn:h+} 
\end{equation}
We further define 
\begin{align}
    \tau_+ = \max(t_c,t_d) .
\end{align}

\begin{figure*}
  \centering
  \includegraphics[width=0.95\textwidth]{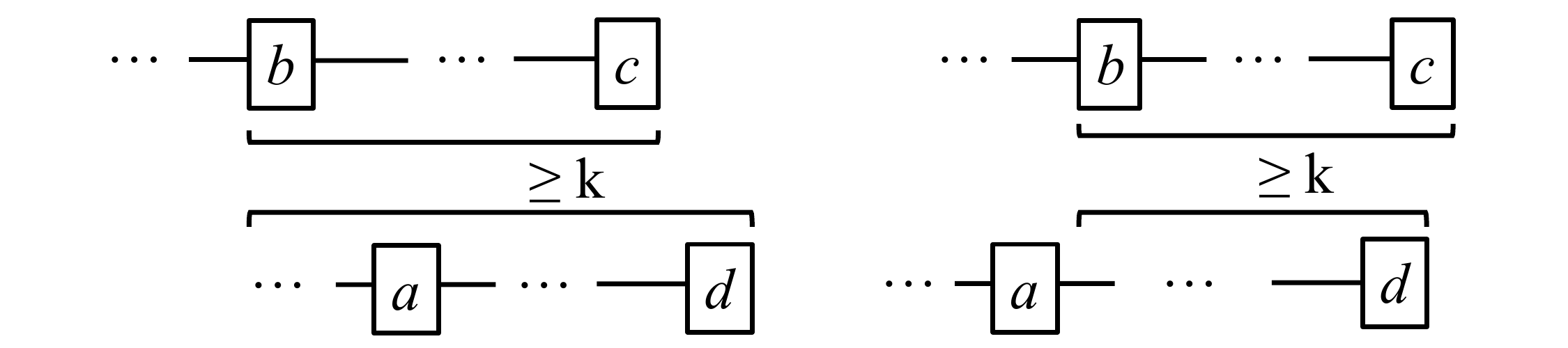}
    \caption{Illustrations of blockchains $c$ and $d$.  Illustrations of the two cases in the proof of Lemma~\ref{lemma:hah+}: $h_a \geq h_b - 1$ (left) and $h_a < h_b - 1$ (right). 
    }
    \label{fig:hah+}
\end{figure*}

Under the hypothetical attack, let block $u$ be the highest honest block mined by
time $\tau$, let $l_{\tau}$ be the lead at time $\tau$, and let block $a$ be the highest block on chain $d$ at time $\tau$. Note that $l_{\tau} + h_u$ is the height of the highest block at time $\tau$ while $a$ is one block at time $\tau$. Thus,
\begin{equation}
l_{\tau} + h_u \geq h_a.
\label{eqn:lt0} 
\end{equation}

Let $H_{t,s}$ be the number of honest blocks mined during $(t, s]$.
Let $A_{t,s}$ be the number of adversarial blocks mined during $(t, s]$.

\begin{lemma}
$H_{\tau,\tau_+-\Delta} \leq h_+ - h_u.$
\label{lemma:honest}
\end{lemma}
\begin{proof}
Since block $u$ is an honest block mined by time $\tau$, all honest blocks mined after $\tau$ will have height at least $h_u+1$ (Condition~\ref{honest_condition}). 
Since one of $c$ or $d$ is credible at time $\tau_+$ and has height at most $h_+$, all honest blocks mined before $\tau_+-\Delta$ must have height at most $h_+$. 
Finally, since honest blocks do not share a height (Condition~\ref{honest_condition}), the lemma is proved.
\end{proof}

\begin{lemma} 
  Every height from $h_a+1$ to $h_+$ contains an adversarial block mined during $(\tau,\tau_+]$.
  \label{lemma:hah+}
\end{lemma}

\begin{proof}
We first show every height from $h_a+1$ to $\min(h_c,h_d)$ contains an adversarial block mined during $(\tau,\tau_+]$.
First, all blocks in blockchain $c$ on height $h_b$ and above are mined after during $(\tau,\tau_+]$; so are all blocks in blockchain $d$ on height $h_a+1$ and above.
We consider two cases as illustrated in Figure~\ref{fig:hah+}.
If $h_a \geq h_b-1$, there exist two blocks on every height from $h_a+1$ to $\min(h_c,h_d)$ that are mined during $(\tau,\tau_+]$.
Given condition~\ref{honest_condition}, one block on every height from $h_a+1$ to $\min(h_c,h_d)$ must be an adversarial block mined during $(\tau,\tau_+]$.
If $h_a < h_b - 1$, by the same argument, one block on every height from $h_b$ to $\min(h_c,h_d)$ must be an adversarial block mined during $(\tau,\tau_+]$.
Furthermore, blocks extending $a$ from height $h_a+1$ to $h_b-1$ do not contain the target transaction $tx$ due to the definition of blockchain $d$.
These must be adversarial blocks because honest nodes would have included $tx$.

Next, we show every height from $\min(h_c,h_d)+1$ to $\max(h_c,h_d)$ also contains an adversarial block mined during $(\tau,\tau_+]$.
If $h_c=h_d$, the statement is vacuous. 
If $h_d>h_c$, we can show that blocks on chain $d$ between height $h_c+1$ to $h_d$ are adversarial blocks. 
Suppose for the sake of contradiction that one of these blocks, say block $f$, is an honest block.
Then block $f$'s parent, denoted as block $e$, must be credible at time $t_e$, does not contain $tx$ up to height $h_b-1$, and has height
\begin{align}
h_e = h_f - 1 \geq h_c \geq h_b + {\cappa-1}.
\end{align}
Thus, $\max(h_c,h_e) = h_+-1$.  This contradicts with $h_+$ being minimized by the pair of blocks $c$ and $d$. 
The case of $h_c > h_d$ is similar.

Thus, we have that every height from $h_a+1$ to $h_+$ contains an adversarial block mined during $(\tau,\tau_+]$. 
\end{proof}

\begin{figure*}
  \centering
  \includegraphics[width=0.95\textwidth]{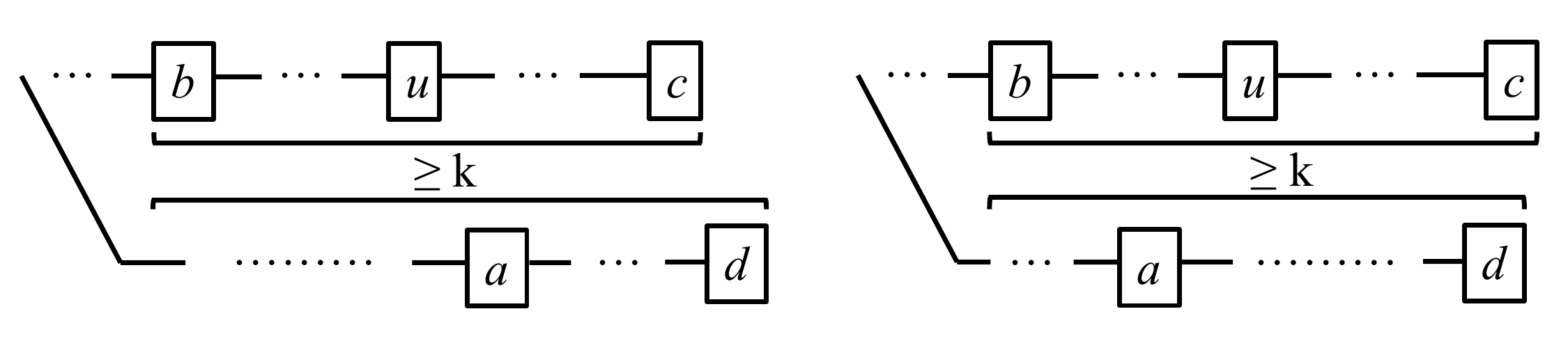}
    \caption{Illustrations of the two cases in the proof of Lemma~\ref{lemma:adversarial}: $h_a \geq h_u$ (left) and $h_a < h_u$ (right). (While block $u$ is included in blockchain $c$ in the graphs, block $u$ can be found in either chain $c$ or chain $d$ or neither. In the $h_a < h_b$ case, though $h_a$ is illustrated to be larger than $h_b$ in the graph, we can have $h_a \leq h_b$.)
    }
    \label{fig:adversarial}
\end{figure*}

\begin{lemma} 
$A_{\tau,\tau_+} \geq h_+ - h_a + \max(h_u - h_b + 1, 0)$.
\label{lemma:adversarial}
\end{lemma}

\begin{proof}
In the case of $h_u - h_b + 1 \leq 0$, this lemma becomes $A_{\tau,\tau_+} \geq h_+ - h_a$, which follows directly from Lemma \ref{lemma:hah+}.
It remains to prove $A_{\tau,\tau_+} \geq (h_+ - h_a) + (h_u - h_b + 1)$ in the case of $h_u \geq h_b$.
In this case, block $b$ and its successors in chain $c$ up to height $h_u$ must be adversarial, because honest nodes will only mine on heights higher than $h_u$ after time $\tau$. 
We consider two cases as illustrated in Figure~\ref{fig:adversarial}.

If $h_a \geq h_u$, then there is an adversarial block mined during $(\tau,\tau_+]$ on every height from $h_b$ to $h_u \leq h_a$, and (due to Lemma~\ref{lemma:hah+}) on every height from $h_a+1$ to $h_+$, giving $A_{\tau,\tau_+} \geq (h_+ - h_a) + (h_u - h_b + 1)$.

If $h_a < h_u$, then the successors of block $a$ in chain $d$ up to height $h_u$ must also be adversarial, because honest nodes will only mine on heights higher than $h_u$ after time $\tau$. 
Also note block $b$ and its successors are not in chain $d$.
Thus, at least $(h_u-h_b+1)+(h_u-h_a)$ adversarial blocks mined during $(\tau,\tau_+]$ on heights no greater than $h_u$. 
Furthermore, at least one adversarial block is mined during $(\tau,\tau_+]$ on every height from $h_u+1$ to $h_+$ (as a consequence of Lemma~\ref{lemma:hah+}). Thus, $A_{\tau,\tau_+} \geq (h_u-h_b+1)+(h_u-h_a)+(h_+-h_u)= (h_+ - h_a) + (h_u - h_b + 1)$ as desired. 
\end{proof}

The following corollary is straightforward from Lemma~\ref{lemma:adversarial}:

\begin{corollary}
$A_{\tau,\tau_+} \geq h_+ - h_a$ and
$A_{\tau,\tau_+} \geq h_+ - h_a + h_u - h_b + 1$.
\label{coro:adversarial}
\end{corollary}

Let us now consider the private-mining attack.
Under this attack, let block $u'$ be the highest honest block at time $\tau$ and let $l'_{\tau}$ be the lead at time $\tau$.
The first honest block after $\tau$ will be on
height $h_{u'}+1$ and it will contain the target transaction.
At time $\tau_+$, the public honest chain has length $h_{u'} + H_{\tau,\tau_+-\Delta}$ and the adversary's private chain has length $h_{u'} + l'_{\tau} + A_{\tau,\tau_+}$.
The private-mining attack succeeds as long as the following two conditions hold:
\begin{enumerate}
    \item[(i)] the adversary's private chain is no shorter than the public honest chain, i.e., $l'_{\tau} + A_{\tau,\tau_+} \geq H_{\tau,\tau_+-\Delta}$;
    \item[(ii)] the adversary's private chain is at least  $\cappa-1$ higher than the target transaction, i.e., $l'_{\tau} + A_{\tau,\tau_+} \geq {\cappa}$.
\end{enumerate}


Both conditions follow directly from the results we just established.
For (i),
\begin{align}
l'_{\tau} + A_{\tau,\tau_+} - H_{\tau,\tau_+-\Delta} &\geq l_{\tau} + (h_+ - h_a) - (h_+ - h_u) \label{eqn:C1} \\
&= l_{\tau} + h_u - h_a  \\
&\geq 0 \label{eqn:C2} 
\end{align}
where \eqref{eqn:C1} is due to Lemma~\ref{lem:maxlead},~Lemma~\ref{lemma:honest}~and~Corollary~\ref{coro:adversarial}, and \eqref{eqn:C2} is due to \eqref{eqn:lt0}.
For (ii),
\begin{align}
l'_{\tau} + A_{\tau,\tau_+}
&\geq l_{\tau} + h_+ - h_a + h_u-h_b+1 \label{eqn:D1} \\
&\ge l_{\tau} + \cappa - h_a + h_u \label{eqn:D2} \\
&\geq {\cappa} \label{eqn:D3}
\end{align}
where \eqref{eqn:D1} is due to Lemma~\ref{lem:maxlead} and Corollary~\ref{coro:adversarial},~\eqref{eqn:D2} is due to \eqref{eqn:h+}, and~\eqref{eqn:D3} is due to~\eqref{eqn:lt0}.
Hence the proof of Theorem~\ref{thm:best_strategy}.
\end{proof}

\section{
The Safety Violation Probability: Proof of Theorem~\ref{th:bounds}}
\label{s:probabilistic}

In this section, we differentiate the notion of {\em blocks mined by honest (resp.\ adversarial) nodes} and {\em honest (resp.\ adversarial) blocks}.
In particular, we modify the canonical model of Section~\ref{sec:model:formal} to define a \emph{rigged} model: We convert selected blocks mined by honest nodes into adversarial blocks, in addition to blocks mined by adversarial nodes.  Only the remaining blocks mined by honest nodes are honest blocks.  We let the adversary be informed of
which
blocks are converted.
The rigged model makes the adversary strictly more powerful because the adversary has the option of behaving honestly with these converted blocks. 

We judiciously select blocks for conversion so that all honest blocks are always on different heights.
Theorem~\ref{thm:best_strategy} now guarantees that the private-mining attack is a best attack in the rigged model.
We then upper bound the safety violation probability of the private-mining attack in the rigged model,
which will serve as an upper bound on any attack in the canonical model.



In order to specify 
blocks for conversion, we introduce the following notions:

\begin{definition}[Tailgaters and laggers]
Let block $j$ be mined at time $t_j$.
If no other block is mined during $(t_j-\Delta, t_j]$, then block $j$ is a lagger; otherwise, block $j$ is a tailgater.
\end{definition}

Now each block has two attributes: whether it is mined by an honest or an adversarial node, and whether it is a tailgater or a lagger. 
Each block is mined by an honest node
with probability $\rho$ and 
by an adversarial node with probability $1-\rho$.
Recall that the inter-arrival times in a Poisson process are 
independent and identically distributed
and follow an exponential distribution with the same rate parameter.
Thus, each block is a lagger with probability $g=e^{-\lambda\Delta}$ and a tailgater with probability $1-g$.
Moreover, whether a block is a lagger or tailgater is independent of 
whether it is mined by an honest or an adversarial node
as well as all other blocks' attributes.

We convert all the tailgaters mined by honest nodes into adversarial blocks.
In other words, only laggers mined by honest nodes are honest blocks; all other blocks are adversarial blocks. 
It is not difficult to see that, in this rigged model,
every block is honest with probability $p=g\rho$ and is adversarial with probability $q = 1-p$, independent of all other blocks.
Hence the sequence of honest/adversarial attributes of all blocks form a memoryless binary process.
Like regular adversarial blocks, the adversary can decide what predecessor block 
a converted block extends, and when it is revealed to honest nodes.
A crucial property is that now every honest block is received by all honest nodes before the next honest block is mined.
Therefore, every honest block will be on
a higher height than all previous honest blocks, i.e., Condition~\ref{honest_condition} is always upheld in the rigged model.


\subsection{The Upper Bound}

Let $F$ denote the event that the private-mining attack violates the target transaction's safety in the rigged model.
Now that the conversions make the private-mining attack a best attack (Theorem~\ref{thm:best_strategy}), it remains to upper bound $\Pr(F)$.  


Let $h$ denote the height of the highest honest block mined by time $\tau$.  The first honest block mined after $\tau$ (on height $h+1$) includes the target transaction $tx$.  Let $L$ denote the lead of the adversary at time $\tau$, so the private chain's height is $h+L$ at $\tau$.

If $L \ge \cappa$, the private chain at time $\tau$ is already high enough to commit the adversarial block on height $h+1$.  
As soon as the honest chain grows by $\cappa$ blocks after $\tau$ to commit $tx$, the adversary releases its private chain to violate its safety. 

If $L < \cappa$, the private chain at $\tau$ is not high enough to commit on height $h+1$.
Let $B$ denote the number of adversarial blocks out of the first $\max(2\cappa-L, 0)$ blocks mined after time $\tau$.
Note that $B$ is non-negative for all possible values of $L$.
If $L<\cappa$ but $L+B\ge\cappa$, then by the time $2\cappa-L$ blocks are mined after $\tau$, the private chain's height $h+L+B\ge h+\cappa$ is sufficient to commit $tx$ at height $h+1$, and it is also no shorter than the honest chain, whose height is
$h + (2\cappa - L - B) \le h + \cappa$.
Again, the adversary has been able to violate $tx$'s safety.

If $L+B<\cappa$, the private chain is still not high enough to commit on height $h+1$ by the time $2\cappa-L$ blocks are mined after $\tau$.  At this point, the public honest chain is also higher than the private chain by at least
\begin{align}
    (h + (2\cappa-L-B)-1) - (h+L+B)
    &= 2(\cappa-L-B) -1
\end{align}
where the ``$-1$'' is because the highest honest block is not necessarily public.  A necessary condition for violating the safety of $tx$ is that the adversary gains enough advantage to make up for this deficit at a later time.  (This is not a sufficient condition because the highest honest block may be public to add one to the actual deficit.)
Let $M$ denote the maximum advantage the adversary ever gains during subsequent mining. 
Specifically, $M$ is the maximum reach of the following simple random walk:
The walk starts at $0$; it increments by one upon the mining of an adversarial block (with probability $q$); and it decrements by one upon the mining of an honest block (with probability $p$).
The necessary condition is then $M\ge 2(\cappa-B-L)-1$.

In summary, the safety violation event $F$ occurs only if one of the following three mutually exclusive events occurs:
\begin{itemize}
    \item $L \ge \cappa$,
    \item $L < \cappa$ but $L+B \ge \cappa$, 
    \item $L+B < \cappa$ but $M \ge 2(\cappa-L-B)-1$. 
\end{itemize}
As $L$, $B$, $M$ are all non-negative, the union of the preceding three events can be concisely written as
\begin{align}
    2L + 2B + M \ge {2\cappa-1}.    
\end{align}
In the remainder of this subsection, we analyze the probability of this event using the joint distribution of $(L,B,M)$.  

By definition, the lead $L$ is determined by block mining times before $\tau$, $B$ is determined by the time $\max(2k-L,0)$ additional blocks are mined after $\tau$, and $M$ is determined by mining times after that.  Due to the memoryless nature of the rigged mining process, $M$ is independent of $(L,B)$.
Under the condition~\eqref{eq:tolerance}, as the maximum reach of the random walk, the distribution of $M$ is geometric~\cite[equation (7.3.5)]{ross1996stochastic}: 
\begin{align} \label{eq:PLi}
    \Pr(M = l) = P_1(l+1;p) 
\end{align}
and
\begin{align} \label{eq:FLi}
    \Pr(M\ge l) = \overline{F}_1(l;p)  
\end{align}
for every $l=0,1,\dots$, where $P_1$ and $\overline{F}_1$ are defined in~\eqref{eq:P1} and~\eqref{eq:F1}.

The lead $L$ is exactly the state of a (continuous-time) birth-death process: It starts at state 0; the mining of each adversarial block corresponds to a birth; and the mining of each honest block corresponds to a death. The process is bounded by zero from below.
It is fair to assume that $\tau$ is not small to allow sufficient mixing in the Markov process, 
so that the distribution of $L$ is identical to the stationary distribution of this birth-death process~\cite[p.~254]{ross1996stochastic} 
under the condition~\eqref{eq:tolerance}.
Interestingly, $L$ has exactly the same geometric distribution as $M$.

Conditioned on $L=l<2\cappa$, $B$ is a binomial random variable with parameters $(2{\cappa}-l,q)$, hence 
\begin{align}
    \Pr(B=j \mid L=l) &= P_2(j; 2{\cappa}-l, q), 
\end{align}
and
\begin{align}
    \Pr(B>j \mid L=l) &= \overline{F}_2(j;2{\cappa}-l,q) \label{eq:PBj}
\end{align}
for every $l=0,1,\dots$, where $P_2$ and $\overline{F}_2$ are defined in~\eqref{eq:P2} and~\eqref{eq:F2}.

Therefore, we have
\begin{align}
    \Pr(F) 
    &\leq \Pr(2L+2B+M \ge 2\cappa-1) \label{eq:FLBM} \\
    &= \Pr(L\ge\cappa)
    + \sum_{l=0}^{\cappa-1} \Pr(L=l) \cdot
        \Pr( 2L+2B+M \ge 2\cappa-1 \mid L = l) \\
    &= \Pr(L\ge\cappa) + \sum_{l=0}^{\cappa-1} \Pr(L=l) \cdot
    \bigg( 
    \Pr(B > \cappa-l-1 \mid L=l) \nonumber \\
    &\qquad +    
    \sum_{j=0}^{\cappa-l-1} \Pr(B=j \mid L = l) \cdot \Pr(M\ge2\cappa-1-2l-2j) \bigg)\\
    &= \overline{F}_1(\cappa;p) + \sum_{l=0}^{\cappa-1} P_1(l+1;p) \cdot \bigg( \overline{F}_2(\cappa-l-1; 2{\cappa}-l,1-p) \nonumber \\ &\quad 
    + \sum_{j=0}^{\cappa-l-1} P_2(j; 2{\cappa}-l,1-p) \cdot  \overline{F}_1(2\cappa-1-2l-2j;p) \bigg)
    \label{eqn:F2} 
\end{align}
which is equal to the upper bound~\eqref{eq:upperbound} in Theorem~\ref{th:bounds} (with $l$ replaced by $i-1$).

\subsection{The Lower Bound}
\label{s:lower}

A lower bound can be obtained by calculating the success probability of the private-mining attack under $\Delta=0$. 
Note that with a zero delay, Condition~\ref{honest_condition} holds without any conversion, so 
the fraction of honest mining power should be the original $\rho$ instead of $p$.
Another material 
difference is that all honest blocks become public immediately.
Let $L'$, $B'$, and $M'$ be identically defined as their counterparts $L$, $B$, and $M$ except that the probability $p$ is replaced with $\rho$. 
Under zero delay,
the event that the private-mining attack violates the safety of the target transaction (which appears at time $\tau$) is exactly $2L'+2B'+M'\ge2\cappa$.
We denote this event as $F_0$.
Assuming again $\tau$ is not small, we have the following exact formula for the probability of safety violation of the target transaction as a function of $\cappa$ and $\rho$:
\begin{align}
    \Pr(F_0)
    &=
    \Pr( 2L'+2B'+M'\ge2\cappa ) \\
    &=
    \overline{F}_1(\cappa;\rho) + \sum_{l=0}^{\cappa-1} P_1(l+1;\rho)
    \cdot \bigg( 
    \overline{F}_2(\cappa-l-1; 2\cappa-l, 1-\rho) \nonumber \\
    & \qquad + \sum_{j=0}^{\cappa-l-1} P_2(j; 2\cappa-l, 1-\rho) \cdot  \overline{F}_1(2\cappa-2l-2j;\rho) \bigg) 
\label{eq:PF0}
\end{align}
which is equal to the lower bound~\eqref{eq:lowerbound} in Theorem~\ref{th:bounds}.

\section{The Safety Violation Probability: Proof of Theorem~\ref{th:exponential}}
\label{s:asymptotics}

To gain additional insights, we evaluate the probability of the safety violation event using simple bounding
techniques.  

\subsection{The Upper Bound}

Recall  
$\Pr(F) \le \Pr( 2L + 2B + M \ge 2\cappa-1 )$ from~\eqref{eq:FLBM}.
The moment generating function (MGF) of the geometric random variables $L$ and $M$ can be expressed as
\begin{align}
    \expect{ e^{\nu L} }
    = \expect{ e^{\nu M} } = \frac{ p - q }{ p - q e^\nu } \label{eq:mgfL}
\end{align}
which holds for every $\nu<\log(p/q)$. 
Conditioned on $L=l$, $B$ is a binomial($2{\cappa}-l,q$) random variable, whose MGF is expressed as
\begin{align}
    \expect{ e^{\nu B} | L=l }
    = \left( p+qe^\nu \right)^{2{\cappa}-l}
\end{align}
for every real number $\nu$.

Using the Chernoff bound~\cite{ross2010first}, we have for every $\nu>0$:
\begin{align}
    \Pr & ( 2L+2B+M \ge 2\cappa-1 ) \nonumber \\
    &= \Pr \left( L+B+\frac M2 \ge {{\cappa-\frac12}} \right) \\
    &\le \expect{ e^{\nu\left(L+B+\frac M2{{-\cappa+\frac12}} \right)} } \\
    &= \expect{ \expect{ e^{\nu(L+B)} \mid L } } \expect{e^{\nu\frac M2}} e^{-\nu \left({{\cappa-\frac12}}\right) } \\
    &= \expect{ e^{\nu L} \left(p+qe^\nu \right)^{2{\cappa}-L} }
     \frac{p-q}{p-qe^{\frac{\nu}2}} e^{-\nu \left({{\cappa-\frac12}}\right)} \\
    &= \left(p+qe^\nu \right)^{2{\cappa}} 
    \frac{p-q}{ p - q \frac{e^\nu}{p+qe^\nu} }
     \frac{p-q}{p-qe^{\frac{\nu}2}} e^{-\nu \left({{\cappa-\frac12}}\right)} 
     \label{eq:Pqv} \\
    &= \left(p+qe^\nu \right)^{2\cappa} e^{-\nu \cappa} \frac{(p-q)^2(p+qe^\nu)  e^{\frac\nu2}} {(p^2-q^2e^\nu)(p-qe^{\frac\nu2}) }.  \label{eq:pqnu}
\end{align}
The exponential coefficient for $\cappa$ is equal to
\begin{align}
    2\log(p+qe^\nu) - \nu .
\end{align}
To yield an asymptotic bound, we note that the tightest exponent for $\cappa$ is obtained with
\begin{align} \label{eq:enu}
    e^\nu = p/q
\end{align}
and thus $p+qe^\nu=2p$.
(Note that~\eqref{eq:enu} guarantees that the MGF~\eqref{eq:mgfL} can be invoked to arrive at~\eqref{eq:Pqv}.)
Using~\eqref{eq:FLBM} and plugging~\eqref{eq:enu} back into~\eqref{eq:pqnu}, we get the upper bound of
\begin{align}
    \Pr(F)
    &\leq \Pr(2L+2B+M \ge {2\cappa-1}) \\
    &\le 2 \, (4pq)^{\cappa}\left( \sqrt{\frac{p}q} + 1 \right)
\end{align}
which is equal to~\eqref{eq:aupper}.

\subsection{The Lower Bound}

According to the discussion in Section~\ref{s:lower}, the safety violation event $F_0$ under the private-mining attack in the case of $\Delta=0$ is expressed exactly as
 $2L'+2B' + M' \ge 2\cappa$.
Evidently, 
\begin{align} \label{eq:PFLB'}
    \Pr(F_0) > \Pr( L'+B' \ge {\cappa} ) .
\end{align}
It is not difficult to see that the larger the pre-mining lead $L'$, the more likely that $L'+B'$ meets the threshold ${\cappa}$, hence for every $l=0,1,\dots,k$,
\begin{align} \label{eq:PB'}
    \Pr( B' \ge {\cappa}-l \mid L'=l )
    \ge
    \Pr( B' \ge {\cappa} \mid L'=0 ) .
\end{align}
Let $B_0$ denote a binomial random variable with parameter $(2{\cappa},q)$.  Equations~\eqref{eq:PFLB'} and~\eqref{eq:PB'} imply that
\begin{align}
    \Pr(F_0) > \Pr( B_0 \ge {\cappa} ) .
\end{align}
Using~\cite[Lemma 4.7.2]{ash1965information}, we can write
\begin{align} \label{eq:PF0d}
    \Pr(F_0)
    > \frac1{\sqrt{{\cappa}}} e^{-2{\cappa}
    d\left( \frac12 \,\|\, \rho \right) }
\end{align}
where 
the relative entropy can be evaluated as
\begin{align}
    d\Big( \frac12 \,\Big\|\, \rho \Big)
    &= - \frac12 \log (4\rho(1-\rho)) .
\end{align}
Hence~\eqref{eq:PF0d} becomes the lower bound~\eqref{eq:alower} in Theorem~\ref{th:exponential}.

\section{Numerical Results}
\label{s:numerical}

\begin{figure}
  \centering
  \includegraphics[width=0.95\textwidth]{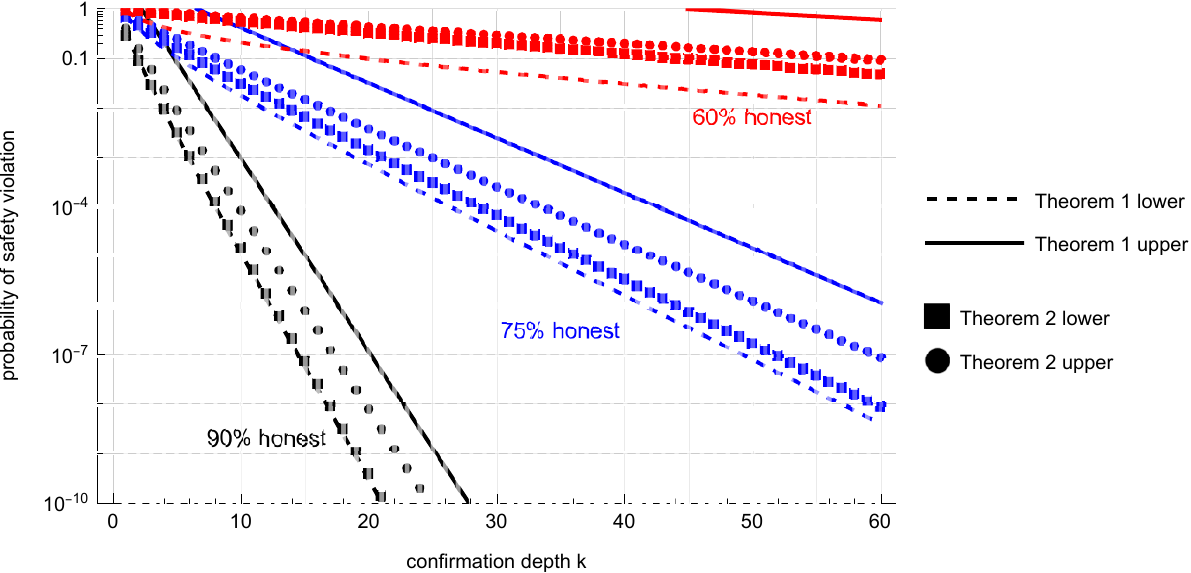}
  \caption{Safety violation probability bounds as a function of the confirmation depth $\cappa$.  All bounds given in Theorems~\ref{th:exponential} and~\eqref{th:bounds} are calculated for Bitcoin's nominal mining rate ($\lambda=1/600$ blocks per second), a block propagation delay bound of $\Delta=10$ seconds, and adversarial mining ratios of 10\% ($\rho=0.9$), 25\% ($\rho=0.75$), and 40\% ($\rho=0.6$).}
  \label{f:btc123}
\end{figure}

\begin{figure}
  \centering
  \includegraphics[width=0.95\textwidth]{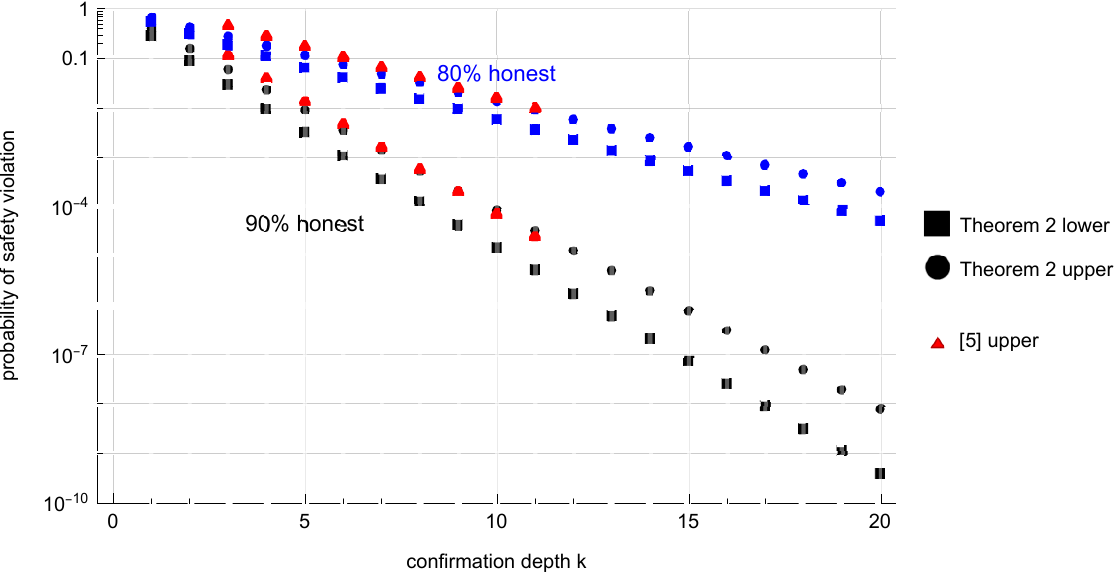}
  \caption{Safety violation probability bounds as a function of the confirmation depth $\cappa$.  The bounds given in Theorem~\eqref{th:bounds} are shown along with the bound of~\cite{gavzi2021practical} for Bitcoin's nominal mining rate ($\lambda=1/600$ blocks per second), a block propagation delay bound of $\Delta=10$ seconds, and adversarial mining ratios of 10\% and 20\%.}
    \label{f:btc12}
    \bigskip\bigskip
\end{figure}

Figure~\ref{f:btc123} shows the upper and lower bounds of Theorems~\ref{th:exponential} and~\ref{th:bounds}
for the Bitcoin mining rate of $\lambda=1/600$ (10 minutes per block), a block propagation delay bound of $\Delta=10$ seconds, and adversarial mining ratios of 10\% ($\rho=0.9$), 25\% ($\rho=0.75$), and 40\% ($\rho=0.6$).
The upper and lower bounds of Theorem~\ref{th:bounds} are quite close for a wide range of honest-to-adversarial mining power ratios.
In addition, the two lower bounds are very close and their gap does not depend on the absolute mining rates and the delay bound.
We also note that the upper bound of Theorem~\ref{th:exponential} can be too loose 
for very small depths.

Consider the practical rule of thumb of committing a Bitcoin block or transaction by 6 confirmations ($\cappa=6$).
If the adversary controls 10\% of the total mining power, the probability of safety violation is between 0.11\% and 0.35\%.  
If the confirmation depth is increased to $\cappa=14$, the probability is then between $2\times10^{-7}$ and $2\times10^{-6}$.  
If the adversary controls 25\% of the total mining power, similar guarantees are obtained at approximately $k=21$ and $k=50$, respectively.


\begin{figure*}
  \centering
  \includegraphics[width=0.99\textwidth]{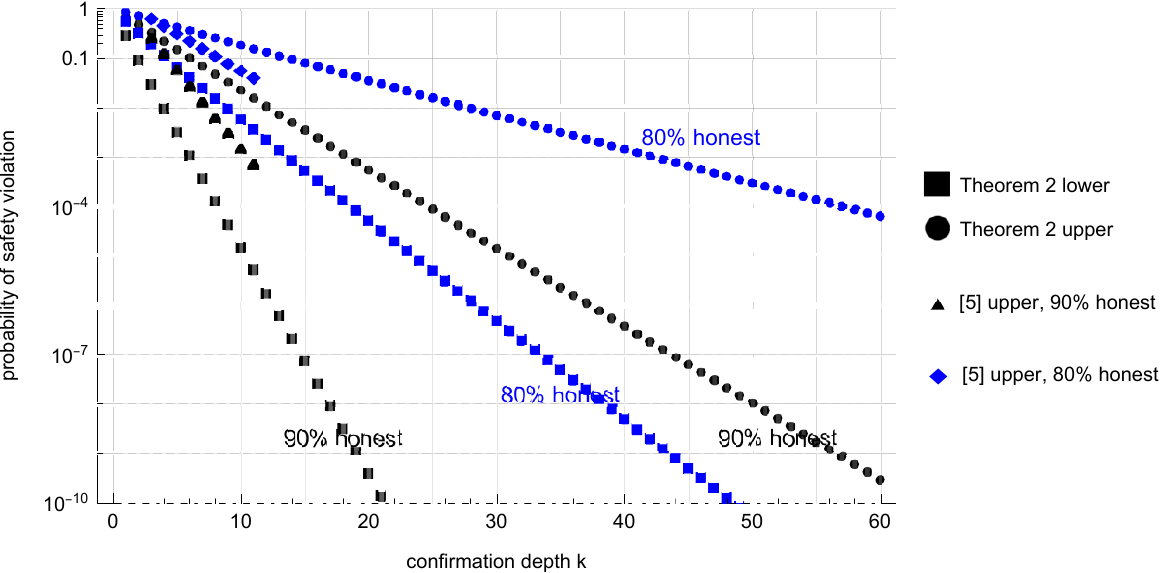}
  \caption{Safety violation probability bounds as a function of the confirmation depth $\cappa$.  The bounds given in Theorem~\ref{th:bounds} are shown along with the bound of~\cite{gavzi2021practical} for Ethereum's nominal mining rate ($\lambda=1/13$ blocks per second), a block propagation delay bound of $\Delta=2$ seconds, and adversarial mining ratios of 10\% and 20\%.}
    \label{f:eth12}
    \vspace{2pt}
\end{figure*}

All the bounds are either exponential or nearly exponential in the confirmation depth.
We also note that the asymptotic slopes of the bounds of Theorem~\ref{th:bounds} match those of Theorem~\ref{th:exponential}.
The more mining power the adversary controls, the flatter the curves are.  For the aforementioned Bitcoin parameters, the fundamental fault tolerance limit is $49.8\%$, i.e., there is no safety regardless of the confirmation depth if $\rho < 0.502$~\cite{dembo2020everything,gavzi2020tight}.

We also show in Figure~\ref{f:btc12} the prior best results from~\cite{gavzi2021practical} for comparison.
For small confirmation depths, our bounds are tighter than theirs despite our much simpler method. 
For example, with 10\% adversary, $\Delta=10$ seconds delay, and 
6-confirmation, our result upper bounds the safety violation probability at 0.35\%, compare to 0.48\% from~\cite{gavzi2021practical}.
As confirmation depths increase, their results eventually become tighter than ours;
the cross-over point for the above parameter occurs at $k=9$. 
With a 20\% adversary, the cross-over point has not occurred within $k \leq 11$. 

Figure~\ref{f:eth12} shows the upper and lower bounds for Ethereum's mining rate of $\lambda=1/13$, a block propagation delay bound of $\Delta=2$ seconds, and adversarial mining ratios of 10\% ($\rho=0.9$) and 20\% ($\rho=0.8$). 
It is easy to produce numerical results for higher adversarial ratios, but we omit them to keep the figure easy to read. 
As expected, our bounds are much looser for Ethereum parameters where the block interval is short relative to the block propagation delay bound. 
The methods from~\cite{gavzi2021practical} yield tighter bounds for such settings. 

Because the lower bounds of both Theorems~\ref{th:exponential} and~\ref{th:bounds} depend only on the ratio of honest and adversarial mining powers, the gap between the two lower bounds remains small regardless of the absolute mining rates and delays. This is clearly seen in Figure~\ref{f:btc123} for the Bitcoin parameters.  (This is also observed for Ethereum's parameters, which we omit here.)

\paragraph{Acknowledgement.}
The authors thank Dr.\ David Tse for helpful discussion. 
This work was supported in part by the National Science Foundation 
under award numbers 2143058 and~2132700.

\bibliographystyle{plain}
\bibliography{refs}

\end{document}